\documentclass[journal,onecolumn]{IEEEtran}
\IEEEoverridecommandlockouts

\usepackage{cite}
\usepackage{amsmath,amssymb,amsfonts}
\usepackage{algorithmic}
\usepackage{graphicx}
\usepackage{textcomp}
\usepackage{xcolor}
\usepackage{mathtools}
\usepackage{bm}
\usepackage{booktabs}
\usepackage[left=2.5cm,right=2.5cm,top=2cm]{geometry}

\newtheorem{theorem}{Theorem}[section]
\newtheorem{lemma}[theorem]{Lemma}
\newtheorem{proposition}[theorem]{Proposition}

\newtheorem{definition}[theorem]{Definition}

\DeclareMathOperator{\Ent}{H}

\def\BibTeX{{\rm B\kern-.05em{\sc i\kern-.025em b}\kern-.08em
    T\kern-.1667em\lower.7ex\hbox{E}\kern-.125emX}}

\begin{document}

\title{Sharper upper bounds for \(q\)-ary\\ and constant-weight \(B_2\) codes}

\author{\IEEEauthorblockN{Stefano Della Fiore}\\
\IEEEauthorblockA{\textit{Department of Information Engineering} \\
\textit{University of Brescia}\\
Brescia, Italy \\
stefano.dellafiore@unibs.it}
}

\maketitle

\begin{abstract}
We derive refined entropy upper bounds for \(q\)-ary \(B_2\) codes by exploiting the Fourier structure of the i.i.d.\ difference distribution \(D=X-Y\). Since the pmf of \(D\) is an autocorrelation, its Fourier series is a nonnegative trigonometric polynomial of degree at most \(q-1\). This leads to a natural convex relaxation over candidate difference distributions, equivalently expressible through an infinite family of positive semidefinite Toeplitz constraints. The resulting formulation admits a simple Gram interpretation and yields certified upper bounds through truncated semidefinite programs. Combined with the prefix--suffix method, this gives improved asymptotic rate upper bounds for \(q\)-ary \(B_2\) codes; in particular, for \(q\in\{9,10,11,12,13\}\) the resulting values improve on the best bounds known in the literature.

We also study binary constant-weight \(B_2\) codes. Extending the distance-distribution method of Cohen, Litsyn, and Z\'emor to the constant-weight setting, and combining it with Litsyn's asymptotic linear-programming bound for constant-weight codes, we derive a new upper bound on the constant-weight \(B_2\) rate.
\end{abstract}


\section{Introduction}

Let \([0,q-1]=\{0,1,\dots,q-1\}\), and let \(C_n\subset[0,q-1]^n\) be a \(q\)-ary code
of length \(n\) with \(|C_n|=M\). We define the asymptotic rate of a family of codes
\(\{C_n\}\) by
\[
R=\limsup_{n\to\infty}\frac{1}{n}\log_q |C_n|,
\]
while logarithms without subscript are taken in base \(2\).

\begin{definition}
A \(q\)-ary code \(C_n=\{c_1,\dots,c_M\}\subset[0,q-1]^n\) is a \(B_2\) code if all sums
\[
c_i+c_j,\qquad 1\le i\le j\le M,
\]
computed over the reals, are distinct.
\end{definition}

For \(q=2\), this is equivalent to the usual notion of a \(\bar 2\)-separable code.
Upper bounds on the rate of \(q\)-ary \(B_2\) codes were obtained in
\cite{Lindstrom69,Lindstrom,CohenLitsynZemor,GuFanMiao,Wang}; see also the references therein.
The constant-weight binary setting was recently investigated in
\cite{SimaLiShomoronyMilenkovic}.

A standard route to asymptotic upper bounds is to relate the code cardinality to the entropy
of suitably chosen random variables associated with prefixes and suffixes of codewords.
For \(q\)-ary \(B_2\) codes, the key local object is the integer difference
\[
D=X-Y,
\]
where \(X\) and \(Y\) are i.i.d.\ \(q\)-ary random variables. The strength of the resulting
rate bound depends on how sharply one can upper bound \(\Ent(D)\). In our earlier work
\cite{DellaFioreDalai}, this entropy term was bounded through a coarse relaxation based only on
the collision constraint
\[
\mathbb{P}(X=Y)=\sum_{a=0}^{q-1}\mathbb{P}(X=a)^2\ge \frac{1}{q},
\]
which led to a closed-form upper bound on \(\Ent(D)\), and hence to an explicit asymptotic
rate bound.

The starting point of the present paper is that the difference distribution \(D=X-Y\) carries
substantially more structure than this single scalar constraint. Its probability mass function
is an autocorrelation, and therefore its Fourier series is a nonnegative trigonometric polynomial
of degree at most \(q-1\). Equivalently, its coefficients satisfy an infinite family of Toeplitz
positive semidefinite constraints. This observation leads to a natural convex relaxation over
candidate difference distributions. The relaxation is strong enough to improve the entropy bound,
yet remains numerically tractable through truncated semidefinite programs (SDPs). At the same time, the
Toeplitz formulation admits a direct Gram interpretation, which makes the origin of the
semidefinite constraints transparent.

This Fourier--Toeplitz viewpoint yields a refinement of the entropy step in the prefix--suffix
method and leads to a new upper bound on the asymptotic rate of \(q\)-ary \(B_2\) codes.
Numerically, the resulting bound improves on the best values previously available to us for
\(q\in\{9,10,11,12,13\}\).

The paper also contains a second result, concerning binary constant-weight \(B_2\) codes of
length \(n\) and weight \(w=\alpha n(1+o(1))\). In this setting the relevant tool is not the
prefix--suffix decomposition, but the distance-distribution method underlying the binary bound of
Cohen, Litsyn, and Z\'emor \cite{CohenLitsynZemor}. We show that this argument admits a
constant-weight refinement. The \(B_2\) property still provides the required combinatorial control
on the distance distribution, while the matching lower estimate is obtained from Litsyn's
asymptotic linear-programming bound for constant-weight codes \cite{Litsyn99}. This leads to a new
upper bound on the asymptotic rate \(r_{B_2}^{\mathrm{cw}}(\alpha)\), improving the recent bound of Sima, Li, Shomorony, and Milenkovi\'c \cite{SimaLiShomoronyMilenkovic}.

The paper is organized as follows. In Sections~II and III we develop the Fourier and Toeplitz
relaxations for the difference distribution and derive the corresponding entropy bound. In
Section~IV we combine this bound with the prefix--suffix argument to obtain a new upper bound for
\(q\)-ary \(B_2\) codes, and in Section~V we report the resulting numerical values. Section~VI is
devoted to the constant-weight binary setting, where we derive a Cohen--Litsyn-type upper bound
and compare it with the recent result of \cite{SimaLiShomoronyMilenkovic}.

\section{Fourier representation of the difference distribution}
\label{sec:fourier-diff}

Let $X$ be a random variable with values in $[0,q-1]$ and pmf $P=(p_0,\dots,p_{q-1})$,
and let $Y$ be an independent copy of $X$. Consider the integer difference
\[
D=X-Y.
\]
Then $D$ takes values in $\{-(q-1),\dots,q-1\}$ and its pmf
$r=(r_{-(q-1)},\dots,r_{q-1})$ is
\begin{equation}
r_k=\mathbb{P}(D=k)=\sum_{i=0}^{q-1} p_i p_{i+k},\qquad k=-(q-1),\dots,q-1,
\label{eq:autocorr}
\end{equation}
with the convention $p_j=0$ for $j\notin\{0,\dots,q-1\}$.

Introduce the characteristic function of $X$,
\[
\phi(\theta)=\mathbb{E}[e^{i\theta X}]
=\sum_{j=0}^{q-1} p_j e^{ij\theta},\qquad \theta\in[0,2\pi],
\]
and the characteristic function of $D$,
\begin{equation}
\psi_D(\theta)=\mathbb{E}[e^{i\theta(X-Y)}]=\phi(\theta)\overline{\phi(\theta)}=|\phi(\theta)|^2.
\label{eq:psiD}
\end{equation}
On the other hand, by Fourier series,
\begin{equation}
\psi_D(\theta)=\sum_{k=-(q-1)}^{q-1} r_k e^{ik\theta}=:R(\theta),
\label{eq:Rtheta}
\end{equation}
where $R(\theta)$ is a trigonometric polynomial of degree at most $q-1$ with real symmetric coefficients
$r_{-k}=r_k$. Combining \eqref{eq:psiD} and \eqref{eq:Rtheta},
\begin{equation}
R(\theta)=|\phi(\theta)|^2\ge 0,\qquad \theta\in[0,2\pi].
\label{eq:R-nonneg}
\end{equation}
Moreover, the coefficients satisfy the inversion formula
\begin{align*}
r_k
&=\frac{1}{2\pi}\int_0^{2\pi} R(\theta)e^{-ik\theta}\,d\theta \nonumber\\
&=\frac{1}{2\pi}\int_0^{2\pi} |\phi(\theta)|^2 e^{-ik\theta}\,d\theta,
\qquad |k|\le q-1.
\end{align*}

In particular,
\[
r_0=\mathbb{P}(X=Y)=\sum_{i=0}^{q-1}p_i^2
=\frac{1}{2\pi}\int_0^{2\pi} |\phi(\theta)|^2\,d\theta.
\]
By Cauchy--Schwarz,
\begin{equation}
r_0=\sum_{i=0}^{q-1}p_i^2\ge \frac{1}{q}.
\label{eq:r0-lb}
\end{equation}

\section{Entropy bounds via Fourier positivity and Toeplitz Gram matrices}
\label{sec:entropy-fourier}

Let $X,Y$ be i.i.d.\ on $[0,q-1]$ with pmf $P$ and let $D=X-Y$. Define
\[
h_q^\star:=\sup \Ent(D),
\]
where the supremum is over all $P$, and $\Ent$ denotes Shannon entropy (in bits), with the convention $0\log 0:=0$.

In \cite{DellaFioreDalai} we used the relaxed bound
\begin{equation}
\Ent(D)\le \bar h_q
:=\Ent\!\Big(\tfrac1q,\underbrace{\tfrac{1}{2q},\dots,\tfrac{1}{2q}}_{2q-2\ \text{times}}\Big)
=\frac{1}{q}\log q+\frac{q-1}{q}\log(2q),
\label{eq:hq-coarse}
\end{equation}
obtained by maximizing entropy under the sole constraint $r_0\ge 1/q$.

\subsection{Fourier-analytic relaxation}

\begin{definition}
For fixed $q\ge 2$ define $h_q^{\mathrm F}$ as the optimal value of
\begin{equation}
\begin{aligned}
\text{maximize}\quad &
-\sum_{k=-(q-1)}^{q-1} r_k\log r_k\\
\text{subject to}\quad
& r_k\ge 0,\quad \sum_{k=-(q-1)}^{q-1} r_k=1,\quad r_{-k}=r_k,\\
& r_0\ge \frac{1}{q},\\
& R(\theta):=\sum_{k=-(q-1)}^{q-1} r_k e^{ik\theta}\ge 0
\quad \text{for all }\theta\in[0,2\pi].
\end{aligned}
\label{eq:fourier-problem}
\end{equation}
\end{definition}

\begin{lemma}
\label{lem:hqF-upper}
For every $q\ge 2$,
\[
h_q^\star\le h_q^{\mathrm F}\le \bar h_q.
\]
\end{lemma}

\begin{IEEEproof}
Any i.i.d.\ difference pmf $r$ satisfies the probability and symmetry constraints, the collision constraint
$r_0\ge 1/q$ by \eqref{eq:r0-lb}, and Fourier nonnegativity $R(\theta)=|\phi(\theta)|^2\ge 0$ by
\eqref{eq:R-nonneg}. Hence it is feasible for \eqref{eq:fourier-problem}, giving $h_q^\star\le h_q^{\mathrm F}$.
Dropping $R(\theta)\ge 0$ enlarges the feasible set and yields \eqref{eq:hq-coarse}, so $h_q^{\mathrm F}\le \bar h_q$.
\end{IEEEproof}

\subsection{Toeplitz formulation and Gram interpretation}

Define $r_k=0$ for $|k|\ge q$. For $N\ge 1$ let the Toeplitz matrix $T^{(N)}(r)$ be
\[
T^{(N)}(r)_{ij}=r_{i-j},\qquad i,j\in\{0,\dots,N-1\}.
\]

Since $r$ comes from an autocorrelation as in \eqref{eq:autocorr}, then Toeplitz positive semidefiniteness is immediate from a Gram construction.
Let $p=(p_0,\dots,p_{q-1})$ and extend it to a length-$(N+q-1)$ vector
$\tilde p=(p_0,\dots,p_{q-1},0,\dots,0)$. For $i\in\{0,\dots,N-1\}$ let $\tilde p^{(i)}$ be the shift of $\tilde p$
by $i$ positions (padding with zeros). Then
\[
\langle \tilde p^{(i)},\tilde p^{(j)}\rangle
=\sum_{t} \tilde p_{t-i}\tilde p_{t-j}
=\sum_{u} p_u p_{u+(i-j)}=r_{i-j}.
\]
Hence $T^{(N)}(r)$ is the Gram matrix of the vectors $\{\tilde p^{(0)},\dots,\tilde p^{(N-1)}\}$ and therefore
$T^{(N)}(r)\succeq 0$. The relaxation keeps the necessary Gram condition while optimizing directly over $r$.

\begin{definition}[Toeplitz SDP formulation]
For fixed $q\ge 2$ define $h_q^{\mathrm{SDP}}$ as the optimal value of
\begin{equation}
\begin{aligned}
\text{maximize}\quad &
-\sum_{k=-(q-1)}^{q-1} r_k\log r_k\\
\text{subject to}\quad
& r_k\ge 0,\quad \sum_{k=-(q-1)}^{q-1} r_k=1,\quad r_{-k}=r_k,\\
& r_0\ge \frac{1}{q},\\
& T^{(N)}(r)\succeq 0\quad \text{for all }N\ge 1.
\end{aligned}
\label{eq:toeplitz-sdp}
\end{equation}
\end{definition}

The correspondence between nonnegative trigonometric polynomials and positive definite Toeplitz forms is classical, see for example \cite{CoverThomas,GrenanderSzego,Dumitrescu}. Hence we get the following equivalence.

\begin{proposition}
\label{prop:hqF-hqSDP}
For every $q\ge 2$ we have $h_q^{\mathrm F}=h_q^{\mathrm{SDP}}$.
\end{proposition}

\section{Improved bounds for $B_2$ codes}
\label{sec:B2}

We now show how the improved entropy bound $h_q^{\mathrm{F}}$ for
$D=X-Y$ enters the derivation of an upper bound on the rate
of $q$-ary $B_2$ codes. The proof below follows the same structure as in
\cite{DellaFioreDalai} and is self-contained.

Let $C_n$ be a $q$-ary $B_2$ code of length $n$ with $M=|C_n|$ codewords.
Split each codeword into a prefix of length $e$ and a suffix of length $f=n-e$,
and partition $C_n$ into classes $P_1,\dots,P_r$ according to the prefix.
Let $M_i=|P_i|$ and denote by $\ell_i\in[0,q-1]^e$ the common prefix of $P_i$.
Thus each class can be written as
\[
P_i=\{(\ell_i,w): w\in W_i\},
\qquad |W_i|=M_i\,,
\]
where $W_i$ is the set of suffixes of the codewords in the subcode $P_i$.

Now, let $(X,Y)$ be a random ordered pair of suffixes obtained by choosing uniformly
an ordered pair of codewords inside the \emph{same} prefix class, and then taking
their suffixes. Then $(X,Y)$ is uniform over $\sum_{i=1}^r M_i^2$ possibilities. Hence
\[
\Ent(X,Y) = \log \sum_{i=1}^r M_i^2.
\]

\begin{lemma}
\label{lem:inj}
If $C_n$ is a $B_2$ code, then the map $\Phi(w_1,w_2)=w_1-w_2$ is injective on
ordered pairs of distinct suffixes (taken within a prefix class).
In other words, the multiset of differences
\[
\{\,w_1-w_2 \mid (w_1,w_2)\in W_i\times W_i,\ w_1\neq w_2,\ i=1,\dots,r\,\}
\]
contains no repetitions.
\end{lemma}

\begin{IEEEproof}
We provide a short self-contained proof in the same spirit as in
\cite{DellaFioreDalai}. Assume, for the sake of contradiction, that there exist
four \emph{distinct} codewords
\[
(\ell_h,w_1),(\ell_h,w_2)\in P_h,
\qquad
(\ell_m,w_3),(\ell_m,w_4)\in P_m,
\]
such that
\begin{equation}
\Phi(w_1,w_2)=\Phi(w_3,w_4),
\label{eq:Phi-eq}
\end{equation}
either with $h=m$ and $(w_1,w_2)\neq (w_3,w_4)$, or with $h\neq m$.

\smallskip\noindent
\emph{Case 1: $h=m$.}
If $\{w_1,w_2\}\cap\{w_3,w_4\}\neq\emptyset$, then \eqref{eq:Phi-eq} either forces
$(w_1,w_2)=(w_3,w_4)$ or contradicts the $B_2$ property of the code $C_n$.

Otherwise $\{w_1,w_2\}\cap\{w_3,w_4\}=\emptyset$.
From \eqref{eq:Phi-eq} we have $w_1-w_2=w_3-w_4$, hence $w_1+w_4=w_2+w_3$.
Therefore
\begin{align}
(\ell_h,w_1)+(\ell_h,w_4)
&=(2\ell_h,\,w_1+w_4)\nonumber\\
&=(2\ell_h,\,w_2+w_3)\nonumber\\
&=(\ell_h,w_2)+(\ell_h,w_3).
\label{eq:sum-eq-sameclass}
\end{align}
The two unordered pairs of codewords
$\{(\ell_h,w_1),(\ell_h,w_4)\}$ and $\{(\ell_h,w_2),(\ell_h,w_3)\}$
are distinct (the four codewords are distinct), yet \eqref{eq:sum-eq-sameclass}
shows they have the same sum, contradicting the $B_2$ property.

\smallskip\noindent
\emph{Case 2: $h\neq m$.}
From \eqref{eq:Phi-eq} we again obtain $w_1+w_4=w_2+w_3$, hence
\begin{align*}
(\ell_h,w_1)+(\ell_m,w_4)
&=(\ell_h+\ell_m,\,w_1+w_4)\\
&=(\ell_h+\ell_m,\,w_2+w_3) \\
&=(\ell_h,w_2)+(\ell_m,w_3).
\end{align*}
Thus two distinct unordered pairs of codewords have the same sum, again
contradicting the $B_2$ property.
In both cases we reach a contradiction.
\end{IEEEproof}

Let $Z:=X-Y\in\mathbb{Z}^f$ denote the componentwise difference of the two
suffixes (over the integers), i.e., $Z=(Z_1,\dots,Z_f)$ with $Z_j=X_j-Y_j$.
From Lemma~\ref{lem:inj} we obtain the entropy decomposition
\begin{equation*}
\Ent(X,Y)=\Ent(Z)+\mathbb{P}(Z=0)\,\Ent(X,Y\mid Z=0).
\end{equation*}

\begin{lemma}
\label{lem:zero-term}
Let $Z=X-Y$ be the suffix difference defined above. Then
\begin{equation*}
\mathbb{P}(Z=0)\,\Ent(X,Y\mid Z=0)\le \frac{r}{M}\log M.
\end{equation*}
In particular, with the choice
\begin{equation}
e=\Big\lfloor \log_q(2M)-\log_q\log M\Big\rfloor,
\qquad r=q^e,
\label{eq:choice-e}
\end{equation}
we have that $\mathbb{P}(Z=0)\,\Ent(X,Y\mid Z=0)\le 2$.
\end{lemma}

\begin{IEEEproof}
Let $S_2:=\sum_{i=1}^r |P_i|^2$. By construction, $(X,Y)$ is uniform over the
$S_2$ ordered pairs of suffixes drawn within the same prefix class, hence each
such ordered pair has probability $1/S_2$.

The event $Z=0$ is equivalent to $X=Y$, which corresponds exactly to the
$M=\sum_{i=1}^r |P_i|$ diagonal ordered pairs (one per codeword). Therefore
\begin{equation}
\mathbb{P}(Z=0)=\mathbb{P}(X=Y)=\frac{M}{S_2}.
\label{eq:PZ0}
\end{equation}
By Cauchy--Schwarz,
\begin{equation*}
S_2=\sum_{i=1}^r |P_i|^2
\ge \frac{\big(\sum_{i=1}^r |P_i|\big)^2}{r}
=\frac{M^2}{r},
\end{equation*}
hence from \eqref{eq:PZ0} we obtain $\mathbb{P}(Z=0)\le r/M$.

Moreover, since all diagonal pairs have the same probability $1/S_2$,
the conditional distribution of $(X,Y)$ given $Z=0$ is uniform over its support,
which has size $M$. Thus
\begin{equation*}
\Ent(X,Y\mid Z=0)=\log M.
\end{equation*}

Finally, with \eqref{eq:choice-e} we have $r=q^e \le 2M/\log M$, and therefore
\[
\mathbb{P}(Z=0)\,\Ent(X,Y\mid Z=0)
\le \frac{r}{M}\log M \le 2\,.
\]
\end{IEEEproof}

Write $Z=(Z_1,\dots,Z_f)$. By subadditivity,
\[
\Ent(Z)\le \sum_{j=1}^f \Ent(Z_j).
\]
Fix $j\in\{1,\dots,f\}$. Let $r^{(j)}$ be the pmf of $Z_j=X_j-Y_j$. Let $I\in\{1,\dots,r\}$ denote the random index of the prefix class from which
the ordered pair of codewords is drawn.
Conditioned on the prefix class index $I=i$, the pair $(X_j,Y_j)$ is i.i.d.\
on $[0,q-1]$, hence the conditional difference pmf $r^{(j)}_{|I=i}$ is feasible
for \eqref{eq:fourier-problem}. Since the feasible set of
\eqref{eq:fourier-problem} is convex and $r^{(j)}=\sum_i \mathbb P(I=i)\,
r^{(j)}_{|I=i}$ is a convex combination of these conditional pmfs, it follows
that $r^{(j)}$ is feasible for \eqref{eq:fourier-problem} as well. Therefore
\[
\Ent(Z_j)\le h_q^{\mathrm F},
\]
and consequently
\begin{equation*}
\Ent(Z)\le \sum_{j=1}^f \Ent(Z_j)\le f\,h_q^{\mathrm F}.
\end{equation*}

We now fix $e$ as in \eqref{eq:choice-e}. Then Lemma~\ref{lem:zero-term} gives
$\mathbb{P}(Z=0)\Ent(X,Y\mid Z=0)\le 2$, and hence
\begin{equation}
\log \sum_{i=1}^r |P_i|^2
= \Ent(X,Y) \le f\,h_q^{\mathrm{F}} + 2.
\label{eq:sumPi-bound}
\end{equation}

On the other hand, the Cauchy--Schwarz inequality yields
\begin{equation}
\frac{M^2}{r} \le \sum_{i=1}^r |P_i|^2.
\label{eq:CS}
\end{equation}
Using~\eqref{eq:CS} and~\eqref{eq:sumPi-bound} we obtain
\[
2\log M - \log r \le f\,h_q^{\mathrm{F}} + 2.
\]

With the choice \eqref{eq:choice-e} we have
\begin{equation*}
\log r = e\log q = \log(2M)-\log\log M + O(1),
\end{equation*}
hence
\begin{equation}
2\log M-\log r = \log M+\log\log M + O(1).
\label{eq:LHS-asymp}
\end{equation}
Moreover, \eqref{eq:choice-e} also gives
\begin{align}
e&=\frac{\log M}{\log q}+O(\log\log M), \nonumber
\\
f&=n-e = n-\frac{\log M}{\log q}+O(\log\log M).
\label{eq:f-asymp}
\end{align}
Substituting \eqref{eq:LHS-asymp}--\eqref{eq:f-asymp} into
$2\log M-\log r \le f\,h_q^{\mathrm F}+2$ and absorbing constants yields
\begin{equation*}
\Big(1+\frac{h_q^{\mathrm F}}{\log q}\Big)\log M
\le n\,h_q^{\mathrm F}+O(\log\log M).
\end{equation*}
Dividing by $n$ and using $\log\log M=o(n)$ (since $M\le q^n$) we obtain
\begin{equation*}
\frac{1}{n}\log M
\le \frac{h_q^{\mathrm F}\log q}{\log q+h_q^{\mathrm F}}+o(1).
\end{equation*}

Therefore we get the following theorem.

\begin{theorem}
\label{thm:main}
For integer $q\ge 2$, let $R_b$ be the asymptotic rate of $q$-ary $B_2$ codes. Then
\begin{equation*}
R_b \le
\frac{h_q^{\mathrm{F}}}{\log q + h_q^{\mathrm{F}}},
\end{equation*}
where $h_q^{\mathrm{F}}$ is the optimal value of the Fourier-analytic
convex program~\eqref{eq:fourier-problem}, equivalently of the Toeplitz
semidefinite program~\eqref{eq:toeplitz-sdp}.
\end{theorem}

\begin{table}[t]
\vspace{0.05in}
\centering
\caption{Numerical comparison of rate upper bounds for $q$-ary $B_2$ codes.
All numbers are rounded upwards.}
\label{tab:numerics}
\setlength{\tabcolsep}{3.5pt}
\renewcommand{\arraystretch}{1.25}
\begin{tabular}{@{}lccccc@{}}
\toprule
 & $q=9$ & $10$ & $11$ & $12$ & $13$ \\
\midrule
$R_q^{\mathrm{new}}$           & \textbf{0.55792} & \textbf{0.55611} & \textbf{0.55457} & \textbf{0.55323} & \textbf{0.55206} \\
Lindstr\"om~\cite{Lindstrom}       & 0.57551 & 0.56839 & 0.56264 & 0.55789 & 0.55390 \\
Ours~\cite{DellaFioreDalai}        & 0.56149 & 0.55966 & 0.55807 & 0.55668 & 0.55545 \\
Wang~\cite{Wang}                   & 0.55841 & 0.55727 & 0.55626 & 0.55536 & 0.55454 \\
\bottomrule
\end{tabular}
\end{table}

\section{Numerical evaluation and comparison with known bounds}
\label{sec:numerics}

We evaluate the bound by solving a truncated Toeplitz SDP, imposing
\(T^{(N)}(r)\succeq 0\) only for \(N\le N_{\max}\). This relaxes the infinite
constraint set and therefore yields a certified upper bound:
\[
h_q^\star \le h_q^{\mathrm F}=h_q^{\mathrm{SDP}}\le h_{q,N_{\max}}.
\]

Using Theorem~\ref{thm:main}, we obtain the explicit rate upper bound
\[
R_b\le R_q^{\mathrm{new}}:=\frac{h_{q,N_{\max}}}{\log q+h_{q,N_{\max}}}.
\]
Table~\ref{tab:numerics} reports values for $N_{\max}=100$ and compares them with
previous upper bounds from Lindstr\"om \cite{Lindstrom}, Wang \cite{Wang}, and our earlier bound \cite{DellaFioreDalai}.

\section{A Cohen--Litsyn upper bound for constant-weight \(B_2\) codes}

Let \(C^{(w)} \subseteq \{0,1\}^n\) be a binary \(B_2\) code of constant weight
\[
w=\alpha n(1+o(1)), \qquad 0<\alpha<1,
\]
and rate
\[
R=\frac{1}{n}\log |C^{(w)}|.
\]
Since all codewords have the same weight, only even distances occur. Let
\[
B^{(w)}_{2i}
=
\frac{1}{|C^{(w)}|}
\bigl|\{(x,y)\in C^{(w)}\times C^{(w)}: d(x,y)=2i\}\bigr|
\]
denote the distance distribution. We write
\[
h(x)=-x\log x-(1-x)\log(1-x)
\]
for the binary entropy function.

As in Cohen et al. \cite{CohenLitsynZemor} and as in the previous section, the \(B_2\) property implies that the map
\[
(x,y)\mapsto x-y
\]
is injective on ordered pairs of distinct codewords. If \(d(x,y)=2i\), then
\(x-y\in\{-1,0,1\}^n\) has exactly \(i\) entries equal to \(+1\), exactly \(i\)
entries equal to \(-1\), and \(n-2i\) zero entries. Hence the number of such
difference vectors is
\[
\binom{n}{i}\binom{n-i}{i}
=
\binom{n}{2i}\binom{2i}{i}.
\]
Therefore
\[
|C^{(w)}|\,B^{(w)}_{2i}\le \binom{n}{i}\binom{n-i}{i},
\]
that is,
\begin{equation*}
B^{(w)}_{2i}\le \frac{\binom{n}{i}\binom{n-i}{i}}{|C^{(w)}|}.
\end{equation*}
Setting \(i=\xi n (1+o(1))\) we get
\begin{equation}\label{eq:cw-B2-upper-entropic}
\frac{1}{n}\log B^{(w)}_{2\xi n}
\le h(2\xi)+2\xi-R+o(1).
\end{equation}

\begin{figure}[t]
    \centering
    \includegraphics[height=250pt]{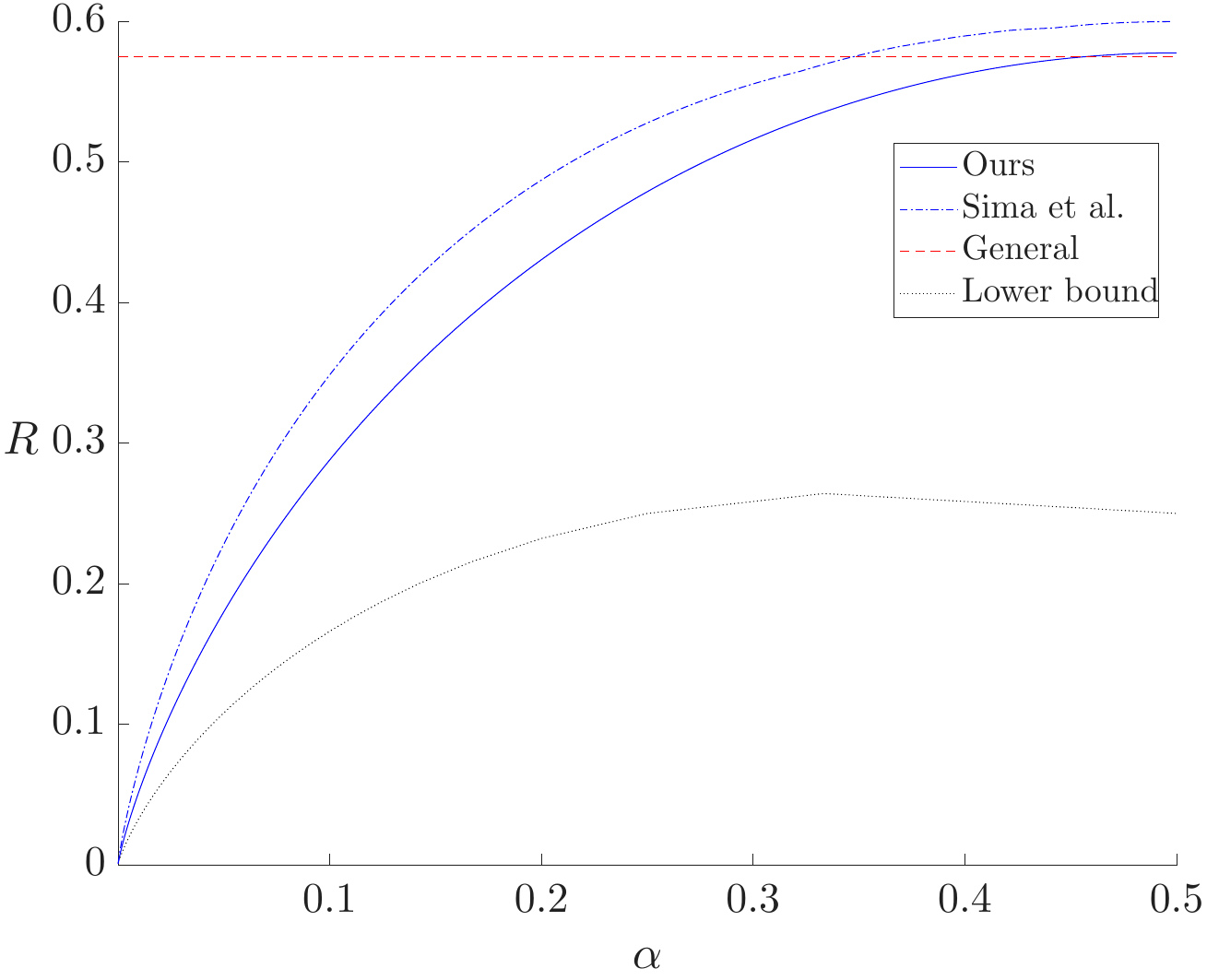}
    \caption{Comparison of upper bounds on the asymptotic rate of binary
    constant-weight \(B_2\) codes. The curve labeled ``Ours'' corresponds to
    Theorem~\ref{thm:cw-B2-LP}, while ``Sima et al.'' denotes the bound
    \(R_S(\alpha)\) from~\cite{SimaLiShomoronyMilenkovic}.
    The upper bound for general codes from~\cite{CohenLitsynZemor} and the constructive lower bound from~\cite{SimaLiShomoronyMilenkovic} are also shown for reference.}
    \label{fig:cw-b2-comparison}
\end{figure}

Let \(\delta_{\mathrm{LP}}(R,\alpha)\) denote the constant-weight linear programming
upper bound on the relative minimum distance, namely
\[
\delta_{\mathrm{LP}}(R,\alpha)
=
\min_{\substack{0\le \beta\le \alpha\\ h(\beta)=R}}
2\,\frac{\alpha(1-\alpha)-\beta(1-\beta)}
{1+2\sqrt{\beta(1-\beta)}}.
\]
By the argument leading to equation (29) of Lemma~3 of \cite{Litsyn99}, one can show there exists an index
\(i=\xi n (1 + o(1))\) with
\[
2\xi\le \delta_{\mathrm{LP}}(R,\alpha)
\]
such that
\begin{equation*}
B^{(w)}_{2\xi n}
\ge
|C^{(w)}|\,
\frac{\binom{w}{\xi n}\binom{n-w}{\xi n}}{\binom{n}{w}}\,
k_{n,w,\xi},
\end{equation*}
where \(k_{n,w,\xi}\) is the polynomial prefactor arising from Litsyn's estimate, in particular, \(k_{n,w,\xi}=n^{-O(1)}\), so it is negligible on the exponential scale. Taking logarithms, we obtain
\begin{equation}\label{eq:litsyn-cw-entropic}
\frac{1}{n}\log B^{(w)}_{2\xi n}
\ge
R-h(\alpha)
+\alpha h\!\left(\frac{\xi}{\alpha}\right)
+(1-\alpha)h\!\left(\frac{\xi}{1-\alpha}\right)
+o(1).
\end{equation}

Combining \eqref{eq:cw-B2-upper-entropic} and \eqref{eq:litsyn-cw-entropic}, we find
\[
2R
\le
h(2\xi)+2\xi+h(\alpha)
-\alpha h\!\left(\frac{\xi}{\alpha}\right)
-(1-\alpha)h\!\left(\frac{\xi}{1-\alpha}\right)
+o(1).
\]
Define
\[
\Psi_\alpha(\xi)
:=
\frac12\Bigl[
h(2\xi)+2\xi+h(\alpha)
-\alpha h\!\left(\frac{\xi}{\alpha}\right)
-(1-\alpha)h\!\left(\frac{\xi}{1-\alpha}\right)
\Bigr].
\]
Then
\begin{equation}\label{eq:implicit-cw-B2}
R\le \Psi_\alpha(\xi)+o(1)
\end{equation}
for some \(\xi\) satisfying \(2\xi\le \delta_{\mathrm{LP}}(R,\alpha)\).

Moreover, \(\Psi_\alpha\) is increasing on \(0\le \xi\le \min\{\alpha,1-\alpha\}\). Indeed,
\[
\Psi_\alpha'(\xi)
=
\log\!\left(
\frac{1-2\xi}{\sqrt{(\alpha-\xi)(1-\alpha-\xi)}}
\right),
\]
and by the AM-GM inequality we have that
\[
\sqrt{(\alpha-\xi)(1-\alpha-\xi)}
\le
\frac{(\alpha-\xi)+(1-\alpha-\xi)}{2}
=
\frac{1-2\xi}{2},
\]
so \(\Psi_\alpha'(\xi)\ge 0\). Since \(2\xi\le \delta_{\mathrm{LP}}(R,\alpha)\), eq. \eqref{eq:implicit-cw-B2}
implies
\begin{equation*}
R\le \Psi_\alpha\!\left(\frac{\delta_{\mathrm{LP}}(R,\alpha)}{2}\right).
\end{equation*}

Taking a subsequence along which the rates converge to
\(r_{B_2}^{\mathrm{cw}}(\alpha)\), and using the continuity of
\(\Psi_\alpha\) and \(\delta_{\mathrm{LP}}(R,\alpha)\), we obtain the following.

\begin{theorem}\label{thm:cw-B2-LP}
For every \(\alpha\in(0,1)\), the asymptotic rate \(r_{B_2}^{\mathrm{cw}}(\alpha)\) of
binary constant-weight \(B_2\) codes of relative weight \(\alpha\) satisfies
\[
r_{B_2}^{\mathrm{cw}}(\alpha)\le \rho(\alpha),
\]
where \(\rho(\alpha)\) is the largest real solution in $R$ of
\[
R=\Psi_\alpha\!\left(\frac{\delta_{\mathrm{LP}}(R,\alpha)}{2}\right).
\]
\end{theorem}

\subsection{Comparison with the best previously known bound}

The best previously known asymptotic upper bound for binary constant-weight
\(B_2\) codes is due to Sima et al. \cite{SimaLiShomoronyMilenkovic}.
In the notation of the present paper, their bound may be written as
\[
r_{B_2}^{\mathrm{cw}}(\alpha)\le R_S(\alpha),
\]
where
\begin{align*}
&R_S(\alpha)
:=
\min_{0\le e \le 1}\;
\max_{\alpha'\in[\max\{0,\alpha-1+e\},\,\min\{e,\alpha\}]} \max\left\{
e h\!\left(\frac{\alpha'}{e}\right),\,
\frac12\left[
e h\!\left(\frac{\alpha'}{e}\right)
+
(1-e) H(p_0,p_1,p_{-1})
\right]
\right\},
\end{align*}
with
\[
p_0=\frac{(\alpha - \alpha')^2+(1-e-\alpha+\alpha')^2}{(1-e)^2},
\quad
p_1=p_{-1}=\frac{1-p_0}{2},
\]
and
\[
H(p_0,p_1,p_{-1})
=
-p_0\log p_0-p_1\log p_1-p_{-1}\log p_{-1}.
\]

Figure~\ref{fig:cw-b2-comparison} compares \(R_S(\alpha)\) with the
upper bound obtained in Theorem~\ref{thm:cw-B2-LP}. The comparison shows that
our bound is uniformly smaller for all \(\alpha\in(0,1/2]\). Hence, Theorem~\ref{thm:cw-B2-LP}
strictly improves the best previously known upper bound throughout the full constant-weight regime under consideration. We also observe that for any $w \leq 0.459n$ we improve the bound for general $B_2$ codes of $0.57525$ obtained by Cohen et al. \cite{CohenLitsynZemor}.

\end{document}